\documentclass[]{elsarticle}

\usepackage{lineno,hyperref}
\usepackage{mathrsfs}
\usepackage{amsfonts}
\usepackage{mathrsfs}
\usepackage{amssymb}
\usepackage{epsfig}
\usepackage{tabularx}
\usepackage{graphicx}
\usepackage{cite}
\usepackage{lettrine}
\usepackage{color}
\usepackage{amsmath}
\usepackage{threeparttable}
\usepackage{algorithm}
\usepackage{algorithmic}
\usepackage{geometry}
\usepackage{multirow}
\usepackage{caption}
\usepackage{subfigure}
\usepackage[subnum]{cases}
\usepackage{booktabs}
\newtheorem{theorem}{Theorem}

\newtheorem{definition}{Definition}
\newtheorem{lemma}{Lemma}

\newtheorem{remark}{Remark}
\newtheorem{proof}{Proof}

\journal{Elsevier}

\begin{document}

\begin{frontmatter}

\title{Model Predictive Control for T-S Fuzzy Markovian Jump Systems Using Dynamic Prediction Optimization}

\author[a,b]{Bin Zhang}
\ead{binzhangusst@163.com}

%
%
\cortext[cor1]{Corresponding author}

\address[a]{Department of Automation, Shanghai Jiao Tong University, Shanghai 200240, China}
\address[b]{Key Laboratory of System Control and Information Processing, Ministry of Education of China, Shanghai 200240, China}

\begin{abstract}
In this paper, the model predictive control (MPC) problem is investigated for the constrained discrete-time Takagi-Sugeno fuzzy Markovian jump systems (FMJSs) under imperfect premise matching rules. To strike a balance between initial feasible region, control performance, and online computation burden, a set of mode-dependent state feedback fuzzy controllers within the frame of dynamic prediction optimizing (DPO)-MPC is delicately designed with the perturbation variables produced by the predictive dynamics. The DPO-MPC controllers are implemented via two stages: at the first stage, terminal constraints sets companied with feedback gain are obtained by solving a ``min-max'' problem; at the second stage, and a set of perturbations is designed felicitously to enlarge the feasible region. Here, dynamic feedback gains are designed for off-line using matrix factorization technique, while the dynamic controller state is determined for online over a moving horizon to gradually guide the system state from the initial feasible region to the terminal constraint set. Sufficient conditions are provided to rigorously ensure the recursive feasibility of the proposed DPO-MPC scheme and the mean-square stability of the underlying FMJS. Finally, the efficacy of the proposed methods is demonstrated through a robot arm system example.
\end{abstract}
\end{frontmatter}
\section{Introduction}\label{sec:1}

With the rapid progress in science and technology, contemporary industrial control systems have grown increasingly intricate, characterized by a proliferation of nonlinearities and uncertainties. The Takagi-Sugeno (T-S) fuzzy modeling approach has emerged as a valuable tool for approximating these complex systems, making T-S fuzzy models integral to the field of control in recent decades. Noteworthy advancements in this area have been documented in various studies \citep{Tran.V21_H,Du.Z21_A,Wang.L18_O,Li.C19_A}. On the other hand, Markovian jump systems (MJSs) have garnered significant attention from the systems science and engineering community due to their efficacy in describing systems experiencing abrupt changes or random fluctuations. By integrating mode variation with T-S fuzzy rules, fuzzy Markovian jump systems (FMJSs) have captured the interest of researchers across disciplines, particularly in the control domain, being a class of stochastic nonlinear dynamic systems. Previous research efforts in this realm include works such as \citep{Xue.M20_E,You.Z20_R,Zeng.P21_E}. It is worth noting that many existing studies predominantly utilize the perfectly matched premises (PMP) framework, also known as parallel distributed compensation (PDC), rather than the imperfectly matched premises (IMP) approach. This preference is primarily attributed to the computational complexity associated with IMP, despite its potential for reduced conservatism, as highlighted in \citep{Lam.H18_A}.

As a cutting-edge modern intelligent control technology, model predictive control (MPC) has demonstrated immense potential in practical applications across various fields \citep{Jin.B20_D,Sun.Z18_R,Zou.Y15_C,Song.Y22_E}. This is primarily due to its significant advantages in efficiently managing optimization problems involving multiple variables and constraints. Numerous research endeavors have been dedicated to addressing MPC challenges in T-S fuzzy systems \citep{Dong.Y20_E,Dong.Y20_E1,Li.F19_F,Hu.J19_O}. For instance, in \citep{Tang.X18_O}, the observer-based output feedback MPC problem was explored for a T-S fuzzy system with data loss. Moreover, in \citep{Teng.L18_E}, a Razumikhin approach was introduced for time-delay fuzzy systems, alongside the provision of two robust MPC algorithms. Regrettably, limited results regarding the MPC problem of FMJSs \citep{Wen.J12_F} have been documented in existing literature, primarily due to the challenges associated with ensuring algorithm feasibility in the simultaneous presence of jump modes and fuzzy rules.

From a practical standpoint, computational burden and performance are consistently pivotal issues for MPC strategies, potentially influencing their integration into industrial engineering applications \citep{Dai.L20_F}, particularly when dealing with a large number of fuzzy controller rules within the IMP category. In the case of online MPC strategies, the continual and substantial computation load over a moving horizon, coupled with the strict requirement for the initial system state to belong to
the terminal constraint set around the origin, pose significant challenges to MPC's practicality. On the other hand, off-line MPC strategies excel in reducing online computational complexity but face stability concerns, especially when dealing with model uncertainties and random variations \citep{Hu.J19_O,Song.Y19_R}. Despite extensive research efforts dedicated to studying both online and off-line MPC strategies, effectively addressing the aforementioned performance issues remains a substantial challenge. Consequently, establishing a comprehensive ``off-line to online" design framework for MPC becomes not only essential but also pragmatic to preserve the strengths of both off-line and online approaches while overcoming their limitations.

To this end, an efficient MPC algorithm proposed in \citep{Kouvaritakis.B00_E,Kouvaritakis.B02_W} strikes a balance between off-line and on-line computation, which requires a lower on-line computation than full on-line MPC strategy in \citep{Song.Y19_R,Zhang.B19_A}. Meanwhile, the convex formulation for dynamic prediction optimizing (DPO) on MPC, as discussed in \citep{Cannon.M05_O,Lu.J19_S}, enhances predictive control by aligning controller state dynamics with system state evolution. Building upon this, in \citep{Nguyen.H20_O}, a new approach was proposed to improve the optimality of DPO-MPC apparently without increasing much online computational load. However, addressing hard constraints in DPO-MPC for complex dynamic systems like FMJSs remains relatively unexplored, motivating the focus of this paper.

In pursuit of enhanced efficiency, a novel MPC algorithm that combines elements of both off-line and online strategies was introduced in \citep{Kouvaritakis.B00_E,Kouvaritakis.B02_W}. This algorithm, requiring less online computation compared to full online MPC strategies \citep{Song.Y19_R,Zhang.B19_A}, offers an additional DoF to extend the initial feasible region. Subsequently, the concept of dynamic prediction optimizing (DPO) within MPC was explored through a convex formulation in works such as \citep{Cannon.M05_O,Lu.J19_S}. This approach empowers the dynamics guiding the predicted controller state to evolve in alignment with the projected system state. Building upon this progress, \citep{Nguyen.H20_O} introduced a novel method to enhance the optimality of DPO-MPC without significantly increasing online computational overhead. However, despite the practical significance of dynamic systems like FMJSs, research on DPO-MPC problems with stringent constraints is still nascent. The identified knowledge gap forms the primary motivation for our current study.

Our objective in this paper is to present a comprehensive ``off-line design and online synthesis" DPO-MPC scheme tailored for a specific  discrete-time stochastic FMJSs with constraints. The primary contributions of this paper are detailed as
follows.  {\it 1)The establishment of an innovative optimizing prediction dynamics framework for MPC design, specifically tailored for T-S FMJSs with hard constraints, marking a pioneering effort in this domain; 2)utilization of mathematical analysis techniques, such as variable substitution, matrix decomposition, and inequality manipulation, to address the non-convexity challenges arising from coupled variables and dynamic state variable introductions; 3)implementation of DPO technology to devise a more adaptable IPM control strategy, effectively mitigating the potential computational overload associated with IPM; 4)development of an ``off-line design and online synthesis'' scheme, as opposed to the complete online MPC scheme \citep{Song.Y19_R,Zhang.B19_A}, aimed at striking a harmonious balance between computational complexity, initial feasible region and control efficacy. This recursive algorithm structure ensures the requisite mean-square stability of the closing-loop MJSs, thus enhancing overall control performance and system stability.}

The remaining sections of the paper are organized as follows. Section \ref{sec:2} presents the formulation of the addressed system model and the DPO-MPC scheme. In Section \ref{sec:3}, the determination of the terminal constraint set and the design of the corresponding control parameters are discussed. The design scheme of perturbation is offered in Section \ref{sec:4} with respect to the ``off-line to online synthesis" approach. A single-link robot arm system example is presented in Section \ref{sec:5} and the paper is concluded in Section \ref{sec:6}.

 \textbf{Notation.} $\mathbb{R}^m$ represents the $m$ dimensional Euclidean space. $\mathcal{I}_{N}$ denotes a sequence from $1$ to $N$, i.e.$\{1,2,\ldots,N\}$. $x_{s+\ell|s}$ and $u_{s+\ell|s}$ stands for the predicted state and control move at the future time instant $s+\ell$ according to the information at the current time instant $s$, respectively, and $\ast_{s|s}\doteq\ast_{s}$. ${\rm col}\{\bullet\}$ and ${\rm diag}\{\bullet\}$ indicate a column block matrix and a diagonal matrix, respectively. $[\,\cdot\,]_{\hbar}$ denotes the $\hbar$th row of a vector.

 \section{Problem Formulation and Preliminaries}\label{sec:2}
\subsection{System Model}\label{sec:2-1}
Let us consider a discrete-time FMJSs using IF-THEN rules:\\
\emph{Plant Rule} $\hbar$: IF $f^{(1)}(x_s)$ is $\mathcal{F}_{\hbar}^{(1)}$, $f^{(2)}(x_s)$ is $\mathcal{F}_{\hbar}^{(2)}$, $\ldots$, and $f^{(r)}(x_s)$ is $\mathcal{F}_{\hbar}^{(r)}$,
\begin{equation} \label{eq:2-1}
\text{THEN} ~~\begin{aligned}
     x_{s+1}=A_{\zeta_{s}, \hbar}x_{s}+B_{\zeta_{s}, \hbar}u_{s}
\end{aligned}
\end{equation}
where $\hbar\in\mathcal{I}_T=\{1,2,\ldots,t\}$, $t$ symbolizes the number of IF-THEN rules;$f(x_s)=[f^{(1)}(x_s),f^{(2)}(x_s),$ $\cdots, f^{(r)}(x_s)]^T$ represents the premise variable of the system and $\{\mathcal{F}_{\hbar}^{(\alpha)}, \alpha=1,2,\ldots,r\}$ represents a fuzzy set; $x_s\in{\mathbb{R}}^{n_x}$ and $u_s\in{\mathbb{R}}^{n_u}$ represent the system state and the system control input, respectively. $A_{\zeta_{s}, \hbar}$ and $B_{\zeta_{s}, \hbar}$ are the known constant matrices with appropriate dimensions.

The stochastic process $\zeta_{s}$ denotes a homogeneous Markov chain in a finite state space $\mathcal{I}_N$ with the transition probability
\begin{align}\label{eq:2-2}
p_{\imath\jmath}=\it{Prob}\big\{\zeta_{s+1}=\jmath|\zeta_{s}=\imath\big\},\quad\forall \imath,\jmath\in{\mathcal{I}_N}
\end{align}
where $0\leq p_{\imath\jmath}\leq 1$. $\mathbb{P}=[p_{\imath\jmath}]$ is called the transition probability matrix.
The initial state $x_{0}$ and mode $\zeta_{0}$ are indicated. For each specific $\zeta_{s}=\imath\in\mathcal{I}_N$, system matrices are represented by $A_{\imath, \hbar}$, $B_{\imath, \hbar}$.

Aided by the T-S fuzzy approach, we establish the overall FMJSs with $\zeta_s=\imath$ as follows:
\begin{equation} \label{eq:2-3}
\begin{aligned}
     x_{s+1}=A_{\imath\theta}x_{s}+B_{\imath\theta}u_{s}
\end{aligned}
\end{equation}
where
\begin{align*}
&A_{\imath\theta}=\sum_{\hbar=1}^t\theta_\hbar(f(x_s))A_{\imath, \hbar},~B_{\imath\theta}=\sum_{\hbar=1}^t\theta_\hbar(f(x_s))B_{\imath, \hbar}\\
&\theta_\hbar(f(x_s))\!=\!\frac{\mathcal{F}_{\hbar}(f(x_s))}{\sum_{\hbar=1}^t\mathcal{F}_{\hbar}(f(x_s))}, \mathcal{F}_{\hbar}(f(x_s))\!=\!\prod_{\alpha=1}^r\!\mathcal{F}_{\hbar}^{(\alpha)}(f^{(\alpha)}(x_s))
\end{align*}
with $\mathcal{F}_{\hbar}^{(\alpha)}(f^{(\alpha)}(x_s))$ referring to the grade of membership of $f^{(\alpha)}(x_s)$ in $\mathcal{F}_{\hbar}^{(\alpha)}$. And $\theta_\hbar(f(x_s))$ implies the standard membership function of rule $\hbar$. For $\forall \hbar\in\mathcal{I}_T$, we have $\theta_\hbar(f(x_s))\geq0$ and $\sum_{\hbar=1}^t\theta_\hbar(f(x_s))=1$. For notational brevity, $\theta_\hbar$ is denoted as $\theta_\hbar(f(x_s))$ for the subsequent analysis.

The FMJSs will be affected by the constraints imposed by engineering practices on inputs and states as follows:
\begin{align}
\big|[u_s]_{e}\big|&\leq[\breve{u}]_e,\;\;\;e\in\mathcal{I}_{n_u}\label{eq:2-4}\\
\big|[\Phi]_fx_s\big|&\leq[\breve{x}]_f,\;\;\;f\in\mathcal{I}_{n_x}\label{eq:2-5}
 \end{align}
where $\breve{u}$ and $\breve{x}$ are the known positive scalars, $\Phi$ is the matrix according to actual needs.
\subsection{DPO-Based Control Law Design}\label{sec:2-2}
Considering the remarkable superiority demonstrated in terms of the applicability complexity, robustness and flexibility to the PMP method, the IPM method is used to design a fuzzy controller for discrete-time FMJSs \eqref{eq:2-3}. However, it is very likely to result in a huge computation burden due to the number of rules for the IPM fuzzy controller.
To alleviate the computational load online and maintain effective system performance, we incorporate a perturbation generated by a dynamic controller into the mode-dependent state feedback fuzzy control law within the DPO-MPC framework. This integration results in the fuzzy predictive dynamic controller law outlined below:\\
\emph{Controller Rule} $\upsilon$: IF $g^{(1)}(x_{s+\ell|s})$ is $\mathcal{H}_{\upsilon}^{(1)}$, $g^{(2)}(x_{s+\ell|s})$ is $\mathcal{H}_{\upsilon}^{(2)}$, $\ldots$, and $g^{(k)}(x_{s+\ell|s})$ is $\mathcal{H}_{\upsilon}^{(k)}$,
\begin{equation}\label{eq:2-6}
\text{THEN} ~~\begin{aligned}
              &u_{s+\ell|s}=K_{\imath, \upsilon}x_{s+\ell|s}+c_{s+\ell|s, \upsilon}, \ell=0,1,2,\ldots\\
              &c_{s+\ell|s,\upsilon}=\mathcal{C}_{\imath,\upsilon}\eta_{s+\ell|s}\\
              &\eta_{s+\ell+1|s}=\mathcal{A}_{\imath,\upsilon}\eta_{s+\ell|s}
              \end{aligned}
\end{equation}
where $\upsilon\in\mathcal{I}_V=\{1,2,\ldots,v\}$ with $v$ symbolizing the number of inference rules; $\{\mathcal{H}_{\upsilon}^{(\kappa)}, \kappa=1,2,\ldots,k\}$ represents a fuzzy set; $g(x_{s+\ell|s})=[g^{(1)}(x_{s+\ell|s}),g^{(2)}(x_{s+\ell|s}),$ $\cdots, g^{(k)}(x_{s+\ell|s})]^T$ represents the premise variable of the controller.

From the fuzzy predictive dynamic controller described by \eqref{eq:2-6}, the mode-dependent fuzzy feedback gain $K_{\imath, \upsilon}$ is designed on terminal constraint set. $\{c_{s+\ell|s,\upsilon}\}$ represents perturbation variables for fine-tuning control inputs, which are generated by the dynamic fuzzy controller. $\eta_{s|s}\triangleq\eta_s\in\mathbb{R}^{n_x}$ represents the dynamic controller state to be determined. The estimator gain $\mathcal{A}_{\imath,\upsilon}$ and dynamic feedback gain $\mathcal{C}_{\imath,\upsilon}$ are needed to be designed. Then, the integrated T-S fuzzy predictive controller with system mode $\zeta_s=i$ is given by
\begin{align}
&u_{s+\ell|s}=K_{\imath\vartheta}x_{s+\ell|\imath}+c_{s+\ell|s, \vartheta}\label{eq:2-7}\\
&c_{s+\ell|s, \vartheta}=\mathcal{C}_{\imath\vartheta}\eta_{s+\ell|s}\label{eq:2-8}\\
&\eta_{s+\ell+1|s}=\mathcal{A}_{\imath\vartheta}\eta_{s+\ell|s}\label{eq:2-9}
\end{align}
where
\begin{align*}
&K_{\imath\vartheta}=\sum_{\upsilon=1}^v\vartheta_\upsilon(g(x_{s+\ell|s}))K_{\imath, \upsilon}, \mathcal{C}_{\imath\vartheta}=\sum_{\upsilon=1}^v\vartheta_\upsilon(g(x_{s+\ell|s}))\mathcal{C}_{\imath,\upsilon}, \mathcal{A}_{i\vartheta}=\sum_{\upsilon=1}^v\vartheta_\upsilon(g(x_{s+\ell|s}))\mathcal{A}_{\imath,\upsilon}, \\ &\vartheta_\upsilon(g(x_{s+\ell|s}))=\frac{\mathcal{H}_{\upsilon}(g(x_{s+\ell|s}))}{\sum_{\upsilon=1}^v\mathcal{H}_{\upsilon}(g(x_{s+\ell|s}))},
\mathcal{H}_{\upsilon}(g(x_{s+\ell|s}))=\prod_{\kappa=1}^k\mathcal{H}_{\upsilon}^{(\kappa)}(g^{(\kappa)}(x_{s+\ell|s})).
\end{align*}
 $\mathcal{H}_{\upsilon}(g(x_{s+\ell|s}))$ refers to the grade of membership of $g^{(\kappa)}(x_{s+\ell|s})$ in $\mathcal{H}_{\upsilon}^{(\kappa)}$. And $\vartheta_\upsilon(g(x_{s+\ell|s}))$ implies the standard membership function of rule $\upsilon$. For $\forall \upsilon\in\mathcal{I}_V$, we have $\vartheta_\upsilon(g(x_{s+\ell|s}))\geq0$ and $\sum_{\upsilon=1}^v\vartheta_\upsilon(g(x_{s+\ell|s}))=1$. For notational brevity,  $\vartheta_\upsilon$ is denoted as its abbreviation of $\vartheta_\upsilon(g(x_{s+\ell|s}))$.
\begin{remark}
The main working principle of MPC is to enlarge the initial feasible region as well as improve the time efficiency of online computing. The initial feasible region is a certain set that contains the initial system state at instant $s=0$, specially, in this paper, it is defined as an ellipsoid set $\Sigma_{x,\zeta_0}$, and thus $x_0\in\Sigma_{x,\zeta_0}$. To this end, a DPO-MPC strategy is introduced in line with \citep{Cannon.M05_O}.
The control input, determined by the combination of $K_{\imath\vartheta}$ and certain perturbations $c_{s+\ell|s, \vartheta}$, aims to guide the system state from the initial feasible region towards the terminal constraint set. Unlike approaches that rely on perturbation sequences of finite length, which are often arbitrarily set, utilizing prediction dynamics for perturbation determination can objectively expand the initial feasible region. This is because the variables involved in the optimization of the initial feasible region under MPC with prediction dynamics, such as $\mathcal{A}_{\imath,\upsilon}$ and $\mathcal{C}_{\imath,\upsilon}$, offer more degrees of freedom compared to those in the effective MPC scheme cited in \citep{Kouvaritakis.B00_E}.
On the other hand, most of variables to determine the controllers are calculated off-line. To be specific, the variables such as such as $K_{\imath,\upsilon}$, $\mathcal{A}_{\imath,\upsilon}$, and $\mathcal{C}_{\imath,\upsilon}$ are calculated off-line, while only $\eta_{s|s}$ needs to be designed online (See from \textbf{OP4}). Thus, the proposed DPO-MPC law will not result in too much online computation burden.
\end{remark}

\subsection{Preliminaries}\label{sec:2-4}
Before delving into the main results, we provide some definitions to facilitate the derivation of subsequent results.
\begin{definition}
The autonomous FMJS $x_{s+1}=A_{\imath\theta}x_{s}$ is considered mean-square stable if, for any initial conditions $x_0\in\Sigma_{x,\zeta_0}$ and $\zeta_{0}\in\mathcal{I}_N$, the following condition is satisfied:
\begin{equation}\label{eq:2-9-1}
 \mathbb{E}_{x_0,\zeta_{0}}\left\{\sum_{s=0}^\infty x_s^Tx_s\right\}<\infty.
\end{equation}
\end{definition}
\begin{definition}\label{de2}
 For the FMJS \eqref{eq:2-3}, if the system state at time instant $s$ belongs to the set $\Sigma_s$ (i.e.~$x_s\in\Sigma_s$), and its future states under admissible fuzzy control also belong to this set (i.e.~$x_{s+\ell}\in\Sigma_s, \ell=1,2,3,\ldots$), then the set $\Sigma_s$ is termed a positive control invariant set.
\end{definition}

The primary goal of this paper is to design a mode-dependent fuzzy predictive dynamic control law for the FMJS \eqref{eq:2-3} subject to hard constraints. Specifically, for any $x_0\in\Sigma_{\zeta_0}$ (referred to as the initial feasible region, which will be defined in subsequent discussions), our aim is to solve an optimization problem at each time instant $s$. This optimization problem seeks to determine the mode-dependent fuzzy feedback gain $K_{\imath, \upsilon}$, estimator gain $\mathcal{A}_{\imath,\upsilon}$, dynamic feedback gain $\mathcal{C}_{\imath,\upsilon}$ and the optimization controller state $\eta_{s}$, ensuring the mean-square stability of the FMJS \eqref{eq:2-3}. This optimization problem is formulated as follows:
\begin{numcases}
 {\textbf{OP1}}\min_{u_{s+\ell|s}, \ell=0,1,2,\ldots}\quad\max_{A_{\imath\theta},B_{\imath\theta},~\imath\in\mathcal{I}_N}\quad J^{\infty}_s\nonumber\\
     \text{s.t.}~x_{s+\ell+1|s}=A_{\imath\theta}x_{s+\ell|s}+B_{\imath\theta}u_{s+\ell|s}, \!\label{eq:3-4-1}\\
     ~~\big|[u_{s+\ell|s}]_{e}\big|\leq[\breve{u}]_e, e\in\mathcal{I}_{n_u}, \!\label{eq:3-4-2}\\
      \big|[\Phi]_fx_{s+\ell|s}\big|\leq[\breve{x}]_f, f\in\mathcal{I}_{n_x},  \!\label{eq:3-4-3}\\
      \mathbb{E}\{\Delta V_{s+\ell|s}\}\leq\!\!-\left(\|x_{s+\ell|s}\|^{2}_{S}\!+\!\|u_{s+\ell|s}\|^{2}_{R}\right)~~~~~~~\!\!\label{adeq:3-4-4}
\end{numcases}
where $ J^{\infty}_s\triangleq\mathbb{E}\Big\{\sum_{\ell=0}^{\infty}\big(\|x_{s+\ell|s}\|^{2}_{S}+\|u_{s+\ell|s}\|^{2}_{R}\big)\Big\}$,  $S>0$ and $R>0$represent known weighting matrices. The term $\mathbb{E}\{\Delta V_{s+\ell|s}\}$  is defined as the expected difference in $V_{s+\ell|s}$, where $V_{s+\ell|s}$ 	
  represents a Lyapunov-like function. We assume that $V_{s+\ell|s}$ satisfies condition \eqref{adeq:3-4-4}, also known as the terminal cost function condition \citep{Zhang.B19_A}. This condition aids in constructing the upper bound of the objective function and achieving mean-square stability for the closed-loop system.

In the subsequent steps, our objective is to devise the fuzzy control law \eqref{eq:2-7}-\eqref{eq:2-9} by addressing {\bf OP1}. Since directly minimizing the cost function over an infinite horizon, especially considering mode jumps and fuzzy rules, can be challenging, we opt to minimize a certain upper bound instead. It's worth noting that the fuzzy control law \eqref{eq:2-7} consists of two components: one related to the mode-dependent fuzzy feedback gain $K_{\imath\vartheta}$ and the other dependent on $c_{s+\ell|s,\vartheta}$ determined by prediction dynamics. Therefore, our approach involves breaking down the optimization problem \textbf{OP1} into several auxiliary optimization problems to achieve our objective.

\section{Optimization Problem in Terminal Constraint Set}\label{sec:3}
\subsection{Control Law in Terminal Constraint Set}\label{sec:3-1}
To begin, we establish the definition of a set known as the terminal constraint set as follows:
\begin{align}\label{eq:3-1}
\Sigma_{x,\imath}\triangleq\left\{x_{s+\ell|s}|x^T_{s+\ell|s}P_{\imath\theta}x_{s+\ell|s}\leq\sigma\right\},
\end{align}
where the scalar $\sigma>0$, the matrix $P_{\imath\theta}\triangleq\sum_{\hbar=1}^t\theta_\hbar P_{\imath,\hbar}$ with $P_{\imath,\hbar}>0, \forall~\imath\in\mathcal{I}_N$. In the terminal constraint set, the mode-dependent state feedback controller is constructed without using perturbation variables by:
\begin{equation}\label{eq:3-2}
u^{\ast}_{s+\ell|s}=K_{\imath\vartheta}x_{s+\ell|s}, ~\imath\in\mathcal{I}_N.
\end{equation}
Applying \eqref{eq:3-2} to FMJS \eqref{eq:2-3} gives rise to
\begin{equation}\label{eq:3-3}
\begin{aligned}
x_{s+\ell+1|s}&=\sum_{\hbar=1}^t\sum_{\upsilon=1}^v\theta_\hbar\vartheta_\upsilon(A_{\imath, \hbar}+B_{\imath, \hbar}K_{\imath, \upsilon})x_{s+\ell|s}\triangleq\tilde{A}_{\imath\theta\vartheta}x_{s+\ell|s}
\end{aligned}
\end{equation}
With regard to the cost function $J^{\infty}_s$, a min-max problem is employed to design a set of mode-dependent fuzzy controllers, outlined as follows:
\begin{align}\label{eq:3-4}
\min_{K_{\imath, \upsilon}, \upsilon=1,2,\ldots,v,~ \imath\in\mathcal{I}_N}\quad\max_{A_{\imath\theta},B_{\imath\theta},~\imath\in\mathcal{I}_N}\quad \hat{J}_s,
\end{align}
where
$\hat{J}_s=\mathbb{E}\Big\{\sum_{\ell=0}^{\infty}\big(\|x_{s+\ell|s}\|^{2}_{S}+\|u^{\ast}_{s+\ell|s}\|^{2}_{R}\big)\Big\}$.

\subsection{Terminal Cost Function}\label{sec:3-2}
Considering the control law \eqref{eq:3-2}, condition \eqref{adeq:3-4-4} in \textbf{OP1} can be reformulated as follows:
\begin{align}\label{eq:3-4-4}
\mathbb{E}\{\Delta V_{s+\ell|s}\}\!\leq\!\!-\left(\|x_{s+\ell|s}\|^{2}_{S}\!+\!\|u^{\ast}_{s+\ell|s}\|^{2}_{R}\right).
\end{align}

We will now endeavor to identify the solvability condition for the terminal constraint \eqref{eq:3-4-4}. To begin, we introduce the following lemma, which is essential for deriving our main results.

\begin{lemma}\label{lm2}
Let the $S$ and $R$ be given positive matrix.  Assume that there exist matrices $W_{\imath,\hbar}>0$, a scalar $\sigma>0$, a set of matrices $\tilde{K}_{\imath, \upsilon}$ and invertible matrices $G_{\imath,\upsilon}$ , such that the following inequalities holds:
\begin{align}\label{eq:3-5}
\begin{bmatrix}
-\mathbb{G}_{\imath,\hbar\upsilon} & \ast & \cdots & \ast & \ast & \ast\\
\sqrt{p_{\imath1}}\Pi_{\imath, \hbar\upsilon} & -W_{1,\lambda} & \cdots  & \ast & \ast & \ast \\
\vdots  &\vdots & \ddots & \ast & \ast & \ast \\
\sqrt{p_{\imath N}}\Pi_{\imath, \hbar\upsilon} & 0 &\cdots & -W_{N,\lambda}& \ast & \ast \\
S&0 &\cdots & 0 & -\sigma S& \ast \\
R\tilde{K}_{\imath, \upsilon} & 0 &\cdots& 0 & 0 & -\sigma R\\
\end{bmatrix}< 0, \forall \imath\in\mathcal{I}_N, \hbar,\lambda\in\mathcal{I}_T
\end{align}
where
\begin{align*}
\mathbb{G}_{\imath,\hbar\upsilon}=G_{\imath,\upsilon}^T+G_{\imath,\upsilon}-W_{\imath,\hbar},~\Pi_{\imath, \hbar\upsilon}=A_{\imath, \hbar}G_{\imath,\upsilon}+B_{\imath, \hbar}\tilde{K}_{\imath, \upsilon}.
\end{align*}
Then, condition \eqref{eq:3-4-4} is satisfied, and the mode-dependent controller gain is calculated by
\begin{align}\label{eq:3-7}
K_{\imath\vartheta}=\sum_{\upsilon=1}^v\vartheta_\upsilon K_{\imath, \upsilon}, K_{\imath, \upsilon}=\tilde{K}_{\imath, \upsilon}G_{\imath,\upsilon}^{-1}.
\end{align}
\end{lemma}
\begin{proof}
Choose a Lyapunov-like function as follows:
\begin{equation*}
V_{s+\ell|s}(x_{s+\ell|s})=x^{T}_{s+\ell|s}P_{\imath\theta}x_{s+\ell|s}, ~\imath\in\mathcal{I}_N.
\end{equation*}
Calculating the difference of $V_{s+\ell|s}$ along system \eqref{eq:3-3} and taking the mathematical expectation yields
\begin{align}\label{eq:3-9}
&\mathbb{E}\{\Delta V_{s+\ell|s}\}\nonumber\\
=&x^{T}_{s+\ell|s}\Big\{\sum^{N}_{\jmath=1}p_{\imath\jmath}\big[\tilde{A}_{\imath\theta\vartheta}^{T}P_{\jmath\theta_+}\tilde{A}_{\imath\theta\vartheta}\big]-P_{\imath\theta}\Big\}x_{s+\ell|s}
\end{align}
where
\begin{align*}
P_{\jmath\theta_+}=\sum_{\lambda=1}^t\theta_{\lambda+} P_{\jmath,\lambda}, ~\theta_{\lambda+}=\theta_\lambda(f(x_{s+\ell+1|s})).
\end{align*}

Notice that $W_{\imath,\hbar}>0$, we have
\begin{align*}
(G_{\imath,\upsilon}-W_{\imath,\hbar}^{-1})^TW_{\imath,\hbar}(G_{\imath,\upsilon}-W_{\imath,\hbar}^{-1})>0
\end{align*}
which leads to
\begin{align}\label{eq:3-10}
G_{\imath,\upsilon}^T+G_{\imath,\upsilon}-W_{\imath,\hbar}<G_{\imath,\upsilon}^TW_{\imath,\hbar}^{-1}G_{\imath,\upsilon}.
\end{align}
Then, keeping \eqref{eq:3-10} in mind, pre- and post-multiplying inequalities \eqref{eq:3-5} by $\text{diag}\{G_{\imath,\upsilon}^{-T}, \mathbb{I}_{n_x}, \cdots, \mathbb{I}_{n_x}, \mathbb{I}_{n_x}, \mathbb{I}_{n_u}\}$ and its transpose leads to
\begin{align}\label{eq:3-11}
\begin{bmatrix}
-W_{\imath,\hbar}^{-1} & \ast & \cdots & \ast & \ast & \ast\\
\sqrt{p_{\imath1}}\tilde{\Pi}_{\imath, \hbar\upsilon} & -W_{1,\lambda} & \cdots  & \ast & \ast & \ast \\
\vdots  &\vdots & \ddots & \ast & \ast & \ast \\
\sqrt{p_{\imath N}}\tilde{\Pi}_{\imath, \hbar\upsilon} & 0 &\cdots & -W_{N,\lambda}& \ast & \ast \\
S&0 &\cdots & 0 & -\sigma S& \ast \\
RK_{\imath, \upsilon} & 0 &\cdots& 0 & 0 & -\sigma R\\
\end{bmatrix}< 0
\end{align}
where
\begin{align*}
K_{\imath, \upsilon}=\tilde{K}_{\imath, \upsilon}G_{\imath,\upsilon}^{-1},~\tilde{\Pi}_{\imath, \hbar\upsilon}=A_{\imath, \hbar}+B_{\imath, \hbar}K_{\imath, \upsilon}.
\end{align*}
For convenience, we define the matrix on the left side of the above inequality \eqref{eq:3-11} as $\mathbb{Z}_{\imath, \hbar\upsilon\lambda}$. In light of standard membership function property of the underlying T-S FMJSs, it is easily seen from \eqref{eq:3-11} that
\begin{align}\label{eq:3-15}
\sum_{\lambda=1}^t\sum_{\hbar=1}^t\sum_{\upsilon=1}^v\theta_{\lambda+}\theta_\hbar\vartheta_\upsilon\mathbb{Z}_{\imath, \hbar\upsilon\lambda}<0.
\end{align}
Then, it follows from \eqref{eq:3-15} that
\begin{align}\label{eq:3-16}
\begin{bmatrix}
-W_{\imath\theta}^{-1} & \ast & \cdots & \ast & \ast & \ast\\
\sqrt{p_{\imath1}}\tilde{A}_{\imath\theta\vartheta} & -W_{1\theta_+} & \cdots  & \ast & \ast & \ast \\
\vdots  &\vdots & \ddots & \ast & \ast & \ast \\
\sqrt{p_{\imath N}}\tilde{A}_{\imath\theta\vartheta} & 0 &\cdots & -W_{N\theta_+}& \ast & \ast \\
S &0 &\cdots & 0 & -\sigma S& \ast \\
RK_{\imath\vartheta} & 0 &\cdots& 0 & 0 & -\sigma R\\
\end{bmatrix}< 0,
\end{align}
where
\begin{align*}
W_{\imath\theta}^{-1}=\sum_{\hbar=1}^t\theta_\hbar W_{\imath,\hbar}^{-1}, ~W_{\jmath\theta_+}=\sum_{\lambda=1}^t\theta_{\lambda+}W_{\jmath,\lambda},~ \imath,\jmath\in\mathcal{I}_N.
\end{align*}
By using the Schur Complement Lemma, we have
\begin{small}
\begin{align}\label{eq:3-17}
-W_{\imath\theta}^{-1}\!+\!\sum^{N}_{\jmath=1}p_{\imath\jmath}\big[\tilde{A}_{\imath\theta\vartheta}^{T}W_{\jmath\theta_+}^{-1}
\tilde{A}_{\imath\theta\vartheta}\big]\!+\!\frac{1}{\sigma}S\!+\!\frac{1}{\sigma}K_{\imath\vartheta}^{T}RK_{\imath\vartheta}<0.
\end{align}
\end{small}
Multiplying both sides of \eqref{eq:3-17} with $\sigma$, substituting $\sigma W_{\imath\theta}^{-1}=P_{\imath\theta}, ~\sigma W_{\jmath\theta_+}^{-1}=P_{\jmath\theta_+}, ~\imath,\jmath\in\mathcal{I}_N$ into \eqref{eq:3-17}, we have
\begin{align}\label{eq:3-18}
-P_{\imath\theta}+\sum^{N}_{\jmath=1}p_{\imath\jmath}\big[\tilde{A}_{\imath\theta\vartheta}^{T}P_{\jmath\theta_+}\tilde{A}_{\imath\theta\vartheta}\big]+S+K_{\imath\vartheta}^{T}RK_{\imath\vartheta}<0.
\end{align}
Based on \eqref{eq:3-9} and \eqref{eq:3-2}, multiplying both sides of \eqref{eq:3-18} with $x^T_{s+\ell|s}$ and its transpose and by using the transposition technique, we can obtain \eqref{eq:3-4-4}. Thus, \eqref{eq:3-4-4} can be guaranteed by \eqref{eq:3-5}, which completes the proof.
\end{proof}

\begin{remark} \label{remk:2}
In order to derive the mode-dependent controllers $K_{\imath\vartheta}$ corresponding to the terminal constraint set $\Sigma_{x,\imath}$, a set of invertible matrices $G_{\imath,\upsilon}$ in regards to the controller fuzzy rules rather than a common invertible matrix $G_{\imath}$ \citep{Xue.M20_E,Xue.M20_H} are introduced into the conditions \eqref{eq:3-10}, which causes less conservatism to the stability condition of the underlying system in the terminal constraint set. In addition, due to $\theta_\hbar(f(x_{s+\ell|s}))\neq\vartheta_\upsilon(g(x_{s+\ell|s}))$ in \eqref{eq:3-3}, the traditional PDC technique cannot be applied. Considering the properties of the membership functions $\theta_\hbar$, $\vartheta_\upsilon$ and $\theta_{\lambda+}$, the information of them is utilized in \eqref{eq:3-15}-\eqref{eq:3-16}.
\end{remark}
\subsection{ Performance Optimization in Terminal Constraint Set}\label{sec:3-3}
In this subsection, we will try to find an upper bound of $\hat{J}_s$ based on \eqref{eq:3-4-4} to design the fuzzy controller in the terminal constraint set.

To ensure the objective remains finite, it is necessary that $x_{\infty|s}=0$, leading to $V_{\infty|s}(x_{\infty|s})=0$. Adding both sides of equation \eqref{eq:3-4-4} from $\ell=0$ to $\infty$ and considering $\lim_{\ell\rightarrow\infty}\mathbb{E}\{V_{s+\ell|s}\}=0$, we obtain
\begin{align}\label{eq:3-19}
\hat{J}_s\leq\mathbb{E}\{V_{s|s}\}=V_{s}=x^{T}_sP_{\imath\theta}x_s,
\end{align}
indicating $\max_{[A_{\imath\theta},B_{\imath\theta}],~\imath\in\mathcal{I}_N} \hat{J}_s\leq x^{T}_sP_{\imath\theta}x_s$.
This sets an upper bound for the objective function across the infinite horizon $\hat{J}_s$. Therefore, our aim regarding the terminal constraint set of MPC is to formulate a series of mode-specific constant control laws, as defined by \eqref{eq:3-2}, to minimize the upper bound $V_{s}$.

Assuming the state at time $s$ lies within the terminal constraint set as defined by \eqref{eq:3-1}, specifically,
\begin{align}\label{ad:3-1}
 x^{T}_sP_{\imath\theta}x_s\leq\sigma.
\end{align}
it is evident from \eqref{eq:3-4-4} that the predicted state at any future time remains within the set $\Sigma_{x,\imath}$ in terms of mean-square, i.e., $x_{s+\ell|s}\in\Sigma_{x,\imath}$. Condition \eqref{ad:3-1} is also referred to as the positive control invariant set condition. Additionally, we derive
\begin{align}\label{adeq:3-2}
 \hat{J}_s\leq\sigma.
\end{align}
Clearly, $\sigma$ serves as an upper bound for $V_{s}$ and consequently for the objective function $\hat{J}_s$.

Next, we will address hard constraints on inputs and states.
\begin{lemma}\label{lm3}
Let the $\breve{u}$ and $\breve{x}$ be given scalars.  Assume that there exist matrices $\mathcal{U}>0$, $\mathcal{X}>0$ and $W_{\imath,\hbar}>0$, a set of matrices $\tilde{K}_{\imath, \upsilon}$ and invertible matrices $G_{\imath,\upsilon}$ for any $\imath\in\mathcal{I}_N$ $\hbar\in\mathcal{I}_T$ and $\upsilon\in\mathcal{I}_V$, such that the following inequalities holds:
\begin{align}
\begin{bmatrix}
-\mathcal{U} & \ast\\
\tilde{K}_{\imath, \upsilon}^T & -\mathbb{G}_{\imath,\hbar\upsilon}\\
\end{bmatrix}&\leq0, \quad
[\mathcal{U}]_{ee}\leq[\breve{u}]_e^2,~e\in\mathcal{I}_{n_u}\label{eq:3-21}\\
 \begin{bmatrix}
-\mathcal{X} & \ast\\
(\Phi W_{\imath,\hbar})^T & -W_{\imath,\hbar}\\
\end{bmatrix}&\leq0, \quad  [\mathcal{X}]_{ff}\leq[\breve{x}]_f^2,~f\in\mathcal{I}_{n_x}\label{eq:3-22}
\end{align}
then hard constraints on  input  $u^{\ast}_{s}$ and state are satisfied, where $[\cdot]_{ff}$$([\cdot]_{ee})$ denotes the $f$th($e$th) diagonal element of ``$\cdot$''.
\end{lemma}

Following the preceding discussions, given condition \eqref{ad:3-1}, we utilize the upper bound $\sigma$ of $\hat{J}_s$ to formulate an optimization problem for deriving the mode-dependent controller gain $K_{\imath\vartheta}$ regarding the terminal constraint set $\Sigma_{x,\imath}$, outlined as:
\begin{align*}
\mathbf{OP2}: &\min_{\substack{W_{\imath,\hbar}, G_{\imath,\upsilon},\tilde{K}_{\imath, \upsilon},\mathcal{U}, \mathcal{X}\\
\imath\in\mathcal{I}_N, \hbar\in\mathcal{I}_T, \upsilon\in\mathcal{I}_V}} \sigma \nonumber\\
\text{s.t.} ~~~~ &\eqref{eq:3-5}, ~\eqref{eq:3-21}, ~\eqref{eq:3-22}.
\end{align*}
\begin{remark}
The controller gain $K_{\imath\vartheta}$ for the terminal constraint set $\Sigma_{x,\imath}$ is computed by solving \textbf{OP2} off-line. Meanwhile, under the positive control invariant condition \eqref{ad:3-1}, the mean-square stability of the closed-loop system \eqref{eq:3-3} is ensured, as proven in \citep{Zhang.B19_A}. However, condition \eqref{ad:3-1} is somewhat conservative since it necessitates the initial presence of the system state within the terminal constraint set, limiting practical application of the MPC strategy. Consequently, in subsequent analysis, efforts are directed towards expanding the feasible state region, termed the initial feasible region. Once the system state resides within this initial region, our objective is to identify admissible control inputs capable of guiding it into the terminal constraint set. This rationale underscores the introduction of perturbation $c_{s+\ell|s, \vartheta}$ in the control law \eqref{eq:2-7}-\eqref{eq:2-9}.
\end{remark}

\section{Perturbation Variable Design}\label{sec:4}
In this section, we present a strategy for designing perturbation variables using an integrated approach from off-line to online settings.

\subsection{Maximizing the Initial Feasible Region}\label{sec:4-1}

Considering equations \eqref{eq:2-3} and \eqref{eq:2-7}-\eqref{eq:2-9}, and defining $\xi_s=[x_s^T ~\eta_s^T]^T\in\mathbb{R}^{2n_x}$, we describe the closed-loop system as follows:
\begin{equation}\label{eq:4-1}
  \left\{\begin{aligned}
  \xi_{s+\ell+1|s}&=\Xi_{\imath\theta\vartheta}\xi_{s+\ell|s}\\
  x_{s+\ell|s}&=\begin{bmatrix}\mathbb{I}_{n_x}  & 0\end{bmatrix}\xi_{s+\ell|s}\\
  u_{s+\ell|s}&=\begin{bmatrix}K_{\imath\vartheta}  & \mathcal{C}_{\imath\vartheta}\end{bmatrix}\xi_{s+\ell|s}
\end{aligned}\right.
\end{equation}
where
\begin{align*}
&\Xi_{\imath\theta\vartheta}\triangleq\begin{bmatrix}\tilde{A}_{\imath\theta\vartheta} & B_{\imath\theta}\mathcal{C}_{\imath\vartheta}\\ 0 & \mathcal{A}_{\imath\vartheta} \end{bmatrix}=\sum_{\hbar=1}^t\sum_{\upsilon=1}^v\theta_\hbar\vartheta_\upsilon\Xi_{\imath, \hbar\upsilon }, \\&\Xi_{\imath, \hbar\upsilon}=\begin{bmatrix}\tilde{\Pi}_{\imath, \hbar\upsilon} & B_{\imath, \hbar}\mathcal{C}_{\imath,\upsilon}\\ 0 & \mathcal{A}_{\imath,\upsilon} \end{bmatrix}.
\end{align*}

Define a set for system \eqref{eq:4-1}  as follows:
\begin{align}\label{eq:4-2}
\Sigma_{\xi,\imath}\triangleq\left\{\xi_{s+\ell|s}|\xi^T_{s+\ell|s}\mathcal{P}_{\imath\vartheta}\xi_{s+\ell|s}\leq1\right\}.
\end{align}
where $\mathcal{P}_{\imath\vartheta}=\sum_{\upsilon=1}^v\vartheta_\upsilon\mathcal{P}_{\imath,\upsilon}$ and $\mathcal{P}_{\imath,\upsilon}$ denotes a collection of symmetric and positive-definite matrices. Subsequently, we demonstrate that $\Sigma_{\xi,\imath}$ serves as a positive control invariant set for \eqref{eq:4-1} and is admissible under constraints \eqref{eq:2-4} and \eqref{eq:2-5}. According to \citep{Song.Y19_R}, the following conditions are formulated to provide theoretical support for a domain attracting constraints for $\xi_{s+\ell|s}$.
\begin{align}
&\Xi_{\imath\theta\vartheta}^T\big(\sum_{\jmath=1}^Np_{\imath\jmath}\mathcal{P}_{\jmath\vartheta_+}\big)\Xi_{\imath\theta\vartheta}-\mathcal{P}_{\imath\vartheta}<0,\label{eq:4-3}\\
&\begin{bmatrix}
-\mathbf{U}& \ast\\
[K_{\imath\vartheta}  ~ \mathcal{C}_{\imath\vartheta}]^T & -\mathcal{P}_{\imath\vartheta}
\end{bmatrix}\leq0, \quad [\mathbf{U}]_{ee}\leq[\breve{u}]_e^2, \quad  e\in\mathcal{I}_{n_u}\label{eq:4-4}\\
&\begin{bmatrix}
-\mathbf{X} & \ast\\
[\Phi ~~0]^T & -\mathcal{P}_{\imath\vartheta}
\end{bmatrix}\leq0, \quad [\mathbf{X}]_{ff}\leq[\breve{x}]_e^2, \quad  f\in\mathcal{I}_{n_x}\label{eq:4-5}
\end{align}
where
\begin{align*}
\mathcal{P}_{\jmath\vartheta_+}=\sum_{\omega=1}^v\vartheta_{\omega+}\mathcal{P}_{\jmath,\omega}, ~\vartheta_{\omega+}=\vartheta_\omega(g(x_{s+\ell+1|s})).
\end{align*}
 The matrices $\mathbf{U}$ and $\mathbf{X}$ are auxiliary and positive.

Up to this point, we formulate an optimization challenge based on \eqref{eq:4-3}-\eqref{eq:4-5} to expand the constraint attraction region. Nevertheless, it is important to note that \eqref{eq:4-3} exhibits non-convexity due to the interaction between $\Xi_{\imath\theta\vartheta}$ and $\mathcal{P}_{\jmath\vartheta_+}$. Therefore, addressing this non-convex issue is essential to ensure solvability of the forthcoming optimization problem. For this purpose, we introduce the following variable transformation.

To demonstrate this, consider $E_{\imath,\upsilon}, ~F_{\imath,\upsilon}\in\mathbb{R}^{n_x\times n_x}$, and positive matrices $M_{\imath,\upsilon}, ~L_{\imath,\upsilon}\in\mathbb{R}^{n_x\times n_x}$ defined by
\begin{align}
&\mathcal{P}_{\imath,\upsilon}=\begin{bmatrix}
                              M_{\imath,\upsilon}^{-1} & M_{\imath,\upsilon}^{-1}E_{\imath,\upsilon}\\
                              E^T_{\imath,\upsilon}M_{\imath,\upsilon}^{-1} & -E^T_{\imath,\upsilon}M_{\imath,\upsilon}^{-1}L_{\imath,\upsilon}F_{\imath,\upsilon}^{-T}\\
                               \end{bmatrix},\label{eq:4-6}\\
&\mathcal{P}_{\imath,\upsilon}^{-1}=\begin{bmatrix}
                              L_{\imath,\upsilon} & F_{\imath,\upsilon}\\
                              F_{\imath,\upsilon}^{T} & -E_{\imath,\upsilon}^{-1}F_{\imath,\upsilon}\\
                               \end{bmatrix}\label{eq:4-7}\\
&X_{\imath,\upsilon}=\mathcal{C}_{\imath,\upsilon}F_{\imath,\upsilon}^{T} , Y_{\imath\jmath,\upsilon\omega}=E_{\jmath,\omega}\mathcal{A}_{\imath,\upsilon}F_{\imath,\upsilon}^{T}.\label{eq:4-8}
\end{align}
Thus, $\mathcal{P}_{\imath,\upsilon}\mathcal{P}_{\imath,\upsilon}^{-1}=\mathbb{I}_{2n_x}$ implies that
\begin{align}\label{eq:4-9}
 E_{\imath,\upsilon}F_{\imath,\upsilon}^T=M_{\imath,\upsilon}-L_{\imath,\upsilon}.
\end{align}
\begin{remark}
 Due to the coupling between variables in \eqref{eq:4-3}-\eqref{eq:4-5}, the method of variable substitution \eqref{eq:4-6}-\eqref{eq:4-8} is adopted to reformulate condition as a convex one for the solvability. This is a natural yet widely used idea of the investigation on MPC problems in the framework of the dynamic output feedback.
\end{remark}

By virtue of \eqref{eq:4-6}-\eqref{eq:4-9}, the following Lemma is presented to help derive the sufficient conditions to \eqref{eq:4-3}-\eqref{eq:4-5}.

\begin{lemma}\label{lm4}
Let the $\tilde{\Pi}_{\imath, \hbar\upsilon}$ be derived by solving \textbf{OP2}.  Assume that there exist matrices $M_{\imath,\upsilon}>0$, $L_{\imath,\upsilon}>0$, $\mathbf{U}>0$ and $\mathbf{X}>0$, a set of matrices $X_{\imath,\upsilon}$ and $Y_{\imath\jmath,\upsilon\omega}$ for any $\imath\in\mathcal{I}_N$, $\hbar\in\mathcal{I}_T$, $\upsilon,\omega\in\mathcal{I}_V$, such that the following inequalities holds:
 \begin{align}
 &\begin{bmatrix}
-\Delta_{\imath,\upsilon}  & \ast & \cdots & \ast \\
\sqrt{p_{\imath1}}\Gamma_{\imath1,\hbar\upsilon\omega}& -\Delta_{1,\omega} & \cdots & \ast\\
\vdots & \vdots & \ddots & \vdots \\
\sqrt{p_{\imath N}}\Gamma_{\imath N,\hbar\upsilon\omega} & 0 & \cdots & -\Delta_{N,\omega}\\
\end{bmatrix}\leq0, \label{eq:4-10}\\
&\begin{bmatrix}
-\mathbf{U}& \ast\\
[K_{\imath,\upsilon}L_{\imath,\upsilon}\!+\!X_{\imath,\upsilon}~K_{\imath,\upsilon}M_{\imath,\upsilon} ]^T & -\Delta_{\imath,\upsilon}
\end{bmatrix}\!\leq\!0,\!~[\mathbf{U}]_{ee}\!\leq\![\breve{u}]_e^2,\label{eq:4-13}\\
&\begin{bmatrix}
-\mathbf{X}& \ast\\
[\Phi L_{\imath,\upsilon} ~\Phi M_{\imath,\upsilon}]^T & -\Delta_{\imath,\upsilon}
\end{bmatrix}\leq0, ~~[\mathbf{X}]_{ff}\leq[\breve{x}]_e^2, \label{eq:4-15}
\end{align}
where
\begin{align*}
\Delta_{\imath,\upsilon}&=\begin{bmatrix}
L_{\imath,\upsilon} & M_{\imath,\upsilon}\\
M_{\imath,\upsilon}  &M_{\imath,\upsilon}\\
\end{bmatrix},  \tilde{\Pi}_{\imath, \hbar\upsilon}=A_{\imath, \hbar}+B_{\imath, \hbar}K_{\imath, \upsilon},
\\
\Gamma_{\imath\jmath,\hbar\upsilon\omega}&=\begin{bmatrix}
                                                           \tilde{\Pi}_{\imath, \hbar\upsilon}L_{\imath,\upsilon}+B_{\imath, \hbar}X_{\imath,\upsilon}    &  \tilde{\Pi}_{\imath, \hbar\upsilon}M_{\imath,\upsilon}\\
                                                            \tilde{\Pi}_{\imath, \hbar\upsilon}L_{\imath,\upsilon}+B_{\imath, \hbar}X_{\imath,\upsilon}+Y_{\imath\jmath,\upsilon\omega}   & \tilde{\Pi}_{\imath, \hbar\upsilon}M_{\imath,\upsilon}\\
                                                        \end{bmatrix}.
\end{align*}
Then, conditions \eqref{eq:4-3}-\eqref{eq:4-5} are satisfied, and mode-dependent estimator gain and dynamic feedback gain are computed through
\begin{align}\label{eq:4-17}
\mathcal{A}_{\imath,\upsilon}=\sum_{\jmath=1}^Np_{\imath\jmath}E_{\jmath,\upsilon}^{-1}Y_{\imath\jmath,\upsilon\upsilon}F_{\imath,\upsilon}^{-T},~\mathcal{C}_{\imath,\upsilon}=X_{\imath,\upsilon}F_{\imath,\upsilon}^{-T}.
\end{align}
\end{lemma}
\begin{proof}
Firstly, according to the property of the standard membership function of the underlying T-S FMJSs, conditions \eqref{eq:4-3}-\eqref{eq:4-5} can be guaranteed by the following inequalities:
\begin{align}
&\Xi_{\imath, \hbar\upsilon }^T\big(\sum_{\jmath=1}^Np_{\imath\jmath}\mathcal{P}_{\jmath,\omega}\big)\Xi_{\imath, \hbar\upsilon }-\mathcal{P}_{\imath,\upsilon}<0,\label{eq:4-18}\\
&\begin{bmatrix}
-\mathbf{U}& \ast\\
[K_{\imath,\upsilon}  ~ \mathcal{C}_{\imath,\upsilon}]^T & -\mathcal{P}_{\imath,\upsilon}
\end{bmatrix}\leq0, \quad [\mathbf{U}]_{ee}\leq[\breve{u}]_e^2, \quad  e\in\mathcal{I}_{n_u},\label{eq:4-19}\\
&\begin{bmatrix}
-\mathbf{X} & \ast\\
[\Phi ~~0]^T & -\mathcal{P}_{\imath,\upsilon}
\end{bmatrix}\leq0, \quad [\mathbf{X}]_{ff}\leq[\breve{x}]_e^2, \quad  f\in\mathcal{I}_{n_x}.\label{eq:4-20}
\end{align}
By applying the Schur Complement Lemma, \eqref{eq:4-18} is valid if and only if
\begin{align}\label{eq:4-21}
\begin{bmatrix}
-\mathcal{P}_{\imath,\upsilon} & \ast & \cdots & \ast \\
\sqrt{p_{\imath1}}\mathcal{P}_{1,\omega}\Xi_{\imath, \hbar\upsilon } & -\mathcal{P}_{1,\omega} & \cdots & \ast\\
\vdots & \vdots & \ddots & \vdots \\
\sqrt{p_{\imath N}}\mathcal{P}_{N,\omega}\Xi_{\imath, \hbar\upsilon } & 0 & \cdots & -\mathcal{P}_{N,\omega}\\
\end{bmatrix}<0,
\end{align}
Subsequently, multiplying both sides of inequality \eqref{eq:4-21} by $\text{diag}\big\{\Theta_{\imath,\upsilon}^T,~\Theta_{1,\omega}^T,~\ldots,~\Theta_{N,\omega}^T\big\}$ and its transpose with
\begin{align*}
\Theta_{\imath,\upsilon}=\begin{bmatrix}
             L_{\imath,\upsilon} & M_{\imath,\upsilon}\\
             F_{\imath,\upsilon}^{T} & 0\\
              \end{bmatrix},
\end{align*}
yields \eqref{eq:4-10}. Therefore, \eqref{eq:4-3} follows from \eqref{eq:4-10}.

Next, multiplying both sides of the first inequality in \eqref{eq:4-19} by $\text{diag}\big\{\mathbb{I}_{n_u},\Theta_{\imath,\upsilon}^T\big\}$ and its transpose gives us \eqref{eq:4-13}. Consequently, \eqref{eq:4-4} follows from \eqref{eq:4-13}.

Finally, transforming the state constraint \eqref{eq:4-20} into \eqref{eq:4-15} can be achieved by multiplying both sides of the first inequality by $\text{diag}\big\{\mathbb{I}_{n_x},\Theta_{\imath,\upsilon}^T\big\}$ and its transpose.
\end{proof}

Before proceeding, we will demonstrate that if the augmented state $\xi$ resides within the set $\Sigma_{\xi,\imath}$, the state $x$ of the system can reach the terminal constraint set $\Sigma_{x,\imath}$ based on conditions \eqref{eq:4-3}-\eqref{eq:4-5}. Consequently, according to Lemma \ref{lm2}, the system state $x$ can be guided towards the equilibrium point using the control input $u^{\ast}_{s}$.

Prior to advancing, the projection of $\Sigma_{\xi,\imath}$ onto the subspace of $x$ is provided as
\begin{align}\label{eq:4-26}
\Sigma_{x\xi,\imath}\triangleq\left\{x_{s+\ell|s}|x^T_{s+\ell|s}L_{\imath\vartheta}^{-1}x_{s+\ell|s}\leq1\right\},
\end{align}
where $L_{\imath\vartheta}=\sum_{\upsilon=1}^v\vartheta_\upsilon L_{\imath,\upsilon}$.

The following Lemma can be derived similarly to \citep{Nguyen.H20_O}.
\begin{lemma} \citep{Nguyen.H20_O} (Thm.2)\label{adlm:1}
Among all solutions satisfying \eqref{eq:4-10}-\eqref{eq:4-15}, the set $\Sigma_{\xi,\imath}$ can be optimized such that $\Sigma_{x,\imath}\subseteq\Sigma_{x\xi,\imath}$.
\end{lemma}

According to Lemma \ref{adlm:1} and condition \eqref{eq:4-10}, it is evident that the predicted future system state $x_{s+\ell|s}$ can enter $\Sigma_{x,\imath}$ using admissible control laws defined by \eqref{eq:4-1}, provided the initial augmented state $\xi_s$ belongs to $\Sigma_{\xi,\imath}$ \citep{Kouvaritakis.B00_E}.

Thus far, the off-line maximization of $\Sigma_{x\xi,\imath}$ across $\mathcal{A}_{\imath\vartheta}$, $\mathcal{C}_{\imath\vartheta}$, $\mathcal{P}_{\imath\vartheta}$, $\mathbf{U}$ and $\mathbf{X}$, subject to \eqref{eq:4-3}-\eqref{eq:4-5}, can be achieved by solving the following optimization problem:
\begin{align*}
\mathbf{OP3}: &\min_{\substack{M_{\imath,\upsilon}, ~L_{\imath,\upsilon}, ~X_{\imath,\upsilon}, ~Y_{\imath\jmath,\upsilon},~\mathbf{U},~\mathbf{X},\\~\imath\in\mathcal{I}_N, \upsilon\in\mathcal{I}_V}} -\log\det L_{\imath,\upsilon}\nonumber\\
&\text{s.t.} ~~~~ \eqref{eq:4-10}-\eqref{eq:4-15}.
\end{align*}
\begin{remark}
Upon obtaining the variables $M_{\imath,\upsilon}$ and $L_{\imath,\upsilon}$, utilize the matrix eigenvalue decomposition condition \eqref{eq:4-9} to determine the respective values of $E_{\imath,\upsilon}$ and $F_{\imath,\upsilon}$. Typically, this decomposition is singular.
 \end{remark}
 \subsection{The Online Optimization Problem of Controller State $\eta_{s}$}\label{sec:4-2}
From the acquired feasible region $\Sigma_{x\xi,\imath}$, within this section, our focus lies in formulating an online optimization problem using the cost function $J^{\infty}_s$. This problem is subject to the initial state requirement (i.e., $\xi_s\in\Sigma_{\xi,\imath}$), ensuring that a series of permissible control strategies derived from this optimization can guide the system state $x$ towards the terminal constraint set.

To attain the stated objective, concerning {\bf OP1}, the current task involves resolving the ``min-max'' problem of the prediction cost
$J^{\infty}_s$ using the DPO input \eqref{eq:2-7} to ascertain a sequence of disturbance variables $\{c_{s+\ell|s, \vartheta}, \ell=0,1,2,\ldots,\infty\}$.  Now, we examine the subsequent quadratic function
 \begin{align}\label{eq:4-26-1}
V_{s+\ell|s}(\xi_{s+\ell|s})=\xi^T_{s+\ell|s}\Psi_{\imath\vartheta}\xi_{s+\ell|s},
 \end{align}
 where $\Psi_{\imath\vartheta}=\sum_{\upsilon=1}^v\vartheta_\upsilon\Psi_{\imath,\upsilon}\in\mathbb{R}^{2n_x\times 2n_x}$, the matrix $\Psi_{\imath,\upsilon}=\text{diag}\{\Psi_{\imath xx,\upsilon},\Psi_{\imath\eta\eta,\upsilon}\}>0$ satisfying
  \begin{align}\label{eq:4-27}
\mathbb{E}\{\Delta V_{s+\ell|s}(\xi_{s+\ell|s})\}\!\leq\!-\left(\|x_{s+\ell|s}\|^{2}_{S}+\|u_{s+\ell|s}\|^{2}_{R}\right).
  \end{align}
Adding both sides of \eqref{eq:4-17} from $\ell=0$ to $\infty$ under the condition $\lim_{\ell\rightarrow\infty}\mathbb{E}\{V_{s+\ell|s}\}=0$, we have $J^{\infty}_s\leq V_{s|s}(\xi_{s|s})$. Then $V_{s|s}(\xi_{s|s})$ is the upper bound of  $J^{\infty}_s$. The following lemma transforms constraint \eqref{eq:4-17}  into a convex expression.
\begin{lemma}
Condition \eqref{eq:4-27} holds true when there exist matrices $\Psi_{\imath xx,\upsilon}>0$ and $\Psi_{\imath\eta\eta,\upsilon}>0$, for any $\imath\in\mathcal{I}_N$, $\upsilon,\omega\in\mathcal{I}_V$,  satisfying the subsequent inequality
\begin{small}
\begin{align}\label{eq:4-28}
\Xi_{\imath, \hbar\upsilon}^T\left(\sum_{\jmath=1}^Np_{\imath\jmath}\Psi_{\jmath,\omega}\right)\Xi_{\imath, \hbar\upsilon}\!-\!\Psi_{\imath,\upsilon}+\mathbb{E}^TS\mathbb{E}+\Lambda_{\imath,\upsilon}^TR\Lambda_{\imath,\upsilon}<0
\end{align}
\end{small}
where $\Lambda_{\imath,\upsilon}=\begin{bmatrix}K_{\imath,\upsilon}  & \mathcal{C}_{\imath,\upsilon}\end{bmatrix}$, $\mathbb{E}=\begin{bmatrix}\mathbb{I}_{n_x}  & 0\end{bmatrix}$.
\end{lemma}

Note that $\Psi_{\imath\vartheta}$ is used to provide an upper bound of $J^{\infty}_s$. Thus, $\Psi_{\imath\vartheta}$ can be calculated by solving the following optimization problem:
\begin{align*}
\mathbf{OP4}: &\min_{\substack{\Psi_{\imath xx,\upsilon},\Psi_{\imath\eta\eta,\upsilon} ,\\\imath\in\mathcal{I}_N, \upsilon,\omega\in\mathcal{I}_V }}&\text{trace}(\Psi_{\imath xx,\upsilon})+\text{trace}(\Psi_{\imath\eta\eta,\upsilon})   \nonumber\\
&~~~~~\text{s.t.} ~~~ \eqref{eq:4-28}.
\end{align*}

Next, we formulate an online problem to determine $\eta_s$. Given that the system state $x_s$ is measured at time instant $s$, the first component of the upper bound, expressed as $x^T_s\Psi_{\imath xx\vartheta}x_s$ is known. This implies that minimizing $V_{s|s}(\xi_{s|s})$ is solely dependent on $\eta^T_s\Psi_{\imath\eta\eta\vartheta}\eta_s$. Thus, we establish the following optimization problem related to $\eta_{s}$: i) ensuring the initial state satisfies $x_0\in\Sigma_{x\xi, \zeta_0}$; ii) determining the necessary perturbations to guide the system into the terminal constraint set $\Sigma_{x,\imath}$.
\begin{align}
\mathbf{OP5}: &\min_{\substack{\eta_{s}}}~~~~ \eta^T_s\Psi_{\imath\eta\eta\vartheta}\eta_s  \nonumber\\
&\text{s.t.} ~~~ \begin{bmatrix}
-1 & \ast \\
\xi_{s} & -\mathcal{P}_{\imath\vartheta}^{-1}\\
\end{bmatrix}\leq0.\label{eq:4-29}
\end{align}

Using the Schur Complement Lemma, it is evident from \eqref{eq:4-29} that $\xi^T_{s}\mathcal{P}_{\imath\vartheta}\xi_{s}\leq1$, ensuring the positive control invariant condition $\xi_s\in\Sigma_{\xi,\imath}$. Notably, at time instant $s$, a perturbation $\eta_{s}$ can be determined by solving an optimization problem based on the current state $x_s$. Subsequently, a series of perturbation variables $\{c_{s+\ell|s,\vartheta}, \ell=0,1,2,\ldots,\infty\}$ are computed using the dynamic predictions \eqref{eq:2-8}-\eqref{eq:2-9}. However, only the initial component $c_{s,\vartheta}$ 	
  influences $u_s$ and affects the plant operation. At the next step, a fresh perturbation is derived from a new optimization based on the updated state $x_{s+1}$. This iterative process continues until the system reaches the terminal constraint set, illustrating the moving horizon optimization principle in MPC.

\subsection{Stability Analysis and Algorithm}\label{sec:4-3}
In this subsection, the following theorem guarantees the solvability of our DPO-MPC algorithm at time step $s>0$, , provided that the optimization problem is solvable at $s=0$,  and it establishes the mean-square stability of the closed-loop system.
\begin{theorem}
Under the conditions that the off-line optimization problems \textbf{OP2}, \textbf{OP3}, and \textbf{OP4} are feasible for FMJS \eqref{eq:2-3}, the online optimization problem \textbf{OP5} remains feasible for all future times, given any initial mode $\zeta_0$ and the initial state $x_0\in\Sigma_{x\xi,\zeta_0}$, This ensures that the system state can be directed into the terminal constraint set $\Sigma_{x,\imath}$ and the designed controller \eqref{eq:2-7} stabilizes the closed-loop system in a mean square sense.
\end{theorem}
\begin{proof}
The proof proceeds in two main steps. Initially, we verify the feasibility of the online optimization problem \textbf{OP5}, ensuring the system state can be guided into the terminal constraint set $\Sigma_{x,\imath}$. Subsequently, we demonstrate the guaranteed stability of the closed-loop system.

1) Recursive feasibility: For \textbf{OP5}, the validity of \eqref{eq:4-29} hinges on the state $\xi_s$. Therefore, to establish the feasibility of \textbf{OP5}, we must demonstrate that \eqref{eq:4-29} holds given $x_0 \in \Sigma_{x\xi,\zeta_0}$. The condition $x_0 \in \Sigma_{x\xi,\zeta_0}$ implies $\xi_0 \in \Sigma_{\xi,\zeta_0}$, ensuring \eqref{eq:4-29} at initial time $s=0$. Consequently, leveraging Lemma \ref{lm4}, $\Sigma_{\xi,\zeta_0}$ serves as an attraction domain for $\xi_{1|0}$. Thus, $\xi_{1|0} = \xi_1$ for some $\Xi_{\imath\theta\vartheta}$, confirming the feasibility of \textbf{OP5} at $s=1$. This procedure extends recursively to subsequent time instances, thereby ensuring the recursive feasibility of \textbf{OP5} under the initial feasibility assumption. Moreover, Lemma \ref{adlm:1} guarantees that a set of permissible control inputs $u_s$ can guide the predicted state $x_{s+\ell|s}$ into the terminal constraint set $\Sigma_{x,\imath}$ given its initial state in $\Sigma_{x\xi,\imath}$. Therefore, due to the guaranteed recursive feasibility, it follows straightforwardly that such permissible control inputs $u_s$ can ultimately steer the state $x_{s}$ into the terminal constraint set $\Sigma_{x,\imath}$.

2) Mean square stability: The system's mean-square stability under fuzzy feedback gain $K_{\imath\vartheta}$ after entry into the terminal constraint set needs to be established. We select a quadratic function candidate $V(x_s) \triangleq x_s^T P_{\zeta_s\theta} x_s$, where $P_{\zeta_s\theta}$ is obtained from \textbf{OP2}. Given that the state $x_{s+1} = x_{s+1|s}$ for some $[A_{\imath\theta}, B_{\imath\theta}, K_{\imath\theta}]$ and feasibility of \textbf{OP2}, it follows from \eqref{eq:3-4-4} that $\mathbb{E}\{\Delta V_{s}\} = \mathbb{E}\{V(x_{s+1})\} - V(x_s) \leq 0$ holds under constraints \eqref{eq:2-4}-\eqref{eq:2-5}. Thus, $\mathbb{E}\{V(x_{s+1})\} - V(x_s) \leq 0$ holds for all $x_s \in \Sigma_{x\xi,\imath}$. Consequently, all states within $\Sigma_{x\xi,\imath}$ satisfy $\lim_{s\rightarrow\infty} \mathbb{E}\big\{\|x_s\|^2\big\} = 0$, ensuring the mean square stability of the closed-loop system.
\end{proof}

To address the couplings within the optimization challenge, we propose a two-part approach: the \emph{Offline} and \emph{Online} segments for DPO-MPC starting from the specified initial state $\zeta_0$.
\begin{table}[!ht]
\label{Alg1}
\begin{center}
 \begin{tabular}{cl}
\toprule
\multicolumn{2}{l}{\normalsize\bf Algorithm~1:~\textbf{\emph{Offline Part}}}\\
\midrule
 {\it  1.} &Begin at time instant $s_0$ and initialize parameters.\\
 {\it  2.} &Resolve optimization issue \textbf{OP2}, utilizing \eqref{eq:3-7} to obtain \\&the mode-dependent controller gain $K_{\imath\vartheta}$ within \\& the terminal constraint set $\Sigma_{x,\imath}$.\\
 {\it  3.} &Address optimization problem \textbf{OP3}.\\
 {\it  4.} &Utilize condition \eqref{eq:4-9} to derive values for $E_{\imath,\upsilon}$ and $F_{\imath,\upsilon}$.\\
 {\it  5.} &Compute $\mathcal{P}_{\imath,\upsilon}$, $\mathcal{P}_{\imath,\upsilon}^{-1}$, $\mathcal{A}_{\imath,\upsilon}$, $\mathcal{C}_{\imath,\upsilon}$ as per \\&\eqref{eq:4-6}-\eqref{eq:4-8} and \eqref{eq:4-17}.\\
 {\it  6.} & Resolve optimization problem \textbf{OP4}.\\
 \toprule
\multicolumn{2}{l}{\normalsize\bf Algorithm~2:~\textbf{\emph{Online Part}}}\\
\midrule
 {\it  1.} &At each time step $s$, verify if $x_s\in\Sigma_{x,\imath}$. If so, compute the control input\\
 &as $u_s=K_{\imath\vartheta}x_s$. Otherwise, determine $\eta_{s}$ by solving\\
 &optimization problem \textbf{OP5}, then calculate the control input as\\
 & $u_{s}=K_{\imath\vartheta}x_{s}+\mathcal{C}_{\imath\vartheta}\eta_{s}$.\\
 {\it  2.} &Apply $u_s$ to the system. Increment $s$ to $s+1$ and return to Step 1.\\
 \bottomrule
\end{tabular}
\end{center}
\end{table}

\section{Illustrative example}\label{sec:5}
\subsection{Example}
In this section, a single-link robot arm system \citep{Xue.M20_E} is used to verify the effectiveness of the control strategy derived in the previous section, in which dynamic systems is presented by
\begin{equation*}
  \ddot{\delta_t}=-\frac{M_{\zeta_t}gL}{J_{\zeta_t}}\sin(\delta_t)-\frac{R}{J_{\zeta_t}}\dot{\delta_t}+\frac{1}{J_{\zeta_t}}u_t
\end{equation*}
where $\delta_t$ represents the angle position of the arm. $g$, $L$, and $R$ are the acceleration of gravity, the arm's length, and the viscous friction's coefficient, respectively. In this simulation, relevant parameters are set as follows: $g=9.81$, $L=0.5$, and $R=2$. The parameter payload mass $M_{\zeta_t}$ and moment of inertia $J_{\zeta_t}$ for the system are presented as: $M_1=J_1=1$,  $M_1=J_1=5$, $M_1=J_1=10$, and $M_1=J_1=15$. We are able to define $\mathbb{P}$ for the Markov jump matrix as
\begin{align*}
\mathbb{P}=
\begin{bmatrix}
0.2 &0.25 &0.4 &0.15\\0.1 &0.2 &0.3 &0.4\\0.3 &0.2 &0.4 &0.1\\0.4 &0.2& 0.2 &0.2\\
\end{bmatrix}
\end{align*}

Define $x_s=[\delta_s, ~\dot{\delta_s}]^T\triangleq[x_s(1) ,~x_s(2)]^T$ as state variables. The single-link robot arm system model can be discretized by the Euler approximation and reconstructed as (Sampling period $T = 0.1$)

\emph{Plant Rule} $1$: IF $x_s(1)$ is about $0$ rad,
\begin{equation*}
\text{THEN} ~\begin{aligned}
     x_{s+1}=\begin{bmatrix} 1& T\\ -\frac{TM_{\zeta_s}gL}{J_{\zeta_s}}& 1-\frac{TR}{J_{\zeta_s}}\end{bmatrix}
     x_{s}+\begin{bmatrix} 0 \\ \frac{T}{J_{\zeta_s}}\end{bmatrix}u_{s}
\end{aligned}
\end{equation*}

\emph{Plant Rule} $2$: IF $x_s(1)$ is about $\pi$ rad or  $-\pi$ rad,
\begin{equation*}
\text{THEN} ~\begin{aligned}
     x_{s+1}=\begin{bmatrix} 1& T\\ -\frac{\beta TM_{\zeta_s}gL}{J_{\zeta_s}}& 1-\frac{TR}{J_{\zeta_s}}\end{bmatrix}
     x_{s}+\begin{bmatrix} 0 \\ \frac{T}{J_{\zeta_s}}\end{bmatrix}u_{s}
\end{aligned}
\end{equation*}
where $\beta=10^{-2}/\pi$, $\zeta_s=1,2,3,4$.

As for this fuzzy model,  the standard membership function is chosen as
\begin{align*}
  &\theta_1(x_s(1))=\left\{\begin{aligned}&\frac{\sin(x_s(1))-\beta x_s(1)}{(1-\beta)x_s(1)},~~ x_s(1)\neq0\\&1, ~~~~~~~~~~~~~~~~~~~~~~~~~~x_s(1)=0\end{aligned}\right.\\
&\theta_2(x_s(1))=1-\theta_1(x_s(1)).
\end{align*}
Set hard constraints on states and inputs as $-\pi\leq x_s(1)\leq\pi$, $-40\leq u_s\leq40$. The initial values of the state and the mode are given by $x_0 = [1.5, -1.5]^T$ , $\zeta_0 = 1$. The weighting
matrices are chosen as $S = \mathbb{I}_2$, $R =0.01$.
\subsection{Solving Off-line Part }
Through the application of \emph{Lemma 1} and \emph{Lemma 2}, the controller gains specific to each operational mode can be derived by solving \textbf{OP2} individually.

\begin{align*}
&K_{1,1}=\begin{bmatrix} -8.1860& -8.7263\end{bmatrix}, ~~~K_{1,2}=\begin{bmatrix} -8.2532 & -9.0878\end{bmatrix}\\
&K_{2,1}=\begin{bmatrix} -21.7988& -25.6431\end{bmatrix}, K_{2,2}=\begin{bmatrix} -23.8203&-31.5950\end{bmatrix}\\
&K_{3,1}=\begin{bmatrix} -20.7920& -17.1278\end{bmatrix}, K_{3,2}=\begin{bmatrix} -16.2692&-21.2610\end{bmatrix}\\
&K_{4,1}=\begin{bmatrix} -17.1860& -12.2843\end{bmatrix}, K_{4,2}=\begin{bmatrix} -10.9302&-13.9034\end{bmatrix}
\end{align*}

Simultaneously, the terminal constraint set $\Sigma_{x,\imath}$ fcorresponding to each mode is depicted by the dotted line in Fig.~\ref{Fig.1}. Subsequently, the expanded initial feasible region $\Sigma_{x\xi,\imath}$ is determined through the resolution of \textbf{OP3}.

The advantages of the proposed DPO-MPC are illustrated through comparative simulations in \citep{Dong.Y20_E1,Zhang.B19_A}. As shown in Fig.~\ref{Fig.1}, we can see that the initial state $x_0$ belongs to $\Sigma_{x\xi,1}$ but outside of $\Sigma_{x,1}$, and the range of the region $\Sigma_{x\xi,1}$ is obviously much wider compared with $\Sigma_{x,1}$. On the other hand, the initial feasible region $\Upsilon_{x\xi,1}$ derived from the EMPC approach exhibits less satisfactory outcomes when contrasted with the DPO-MPC strategy. This is due to constraints that exert considerable influence in specific directions, stemming from the constraints on their optimization flexibility. This underscores the significant enhancement in practical applicability offered by the proposed algorithm.
\begin{figure}[!ht]
 \centering
 {  \begin{minipage}[t]{0.4\textwidth}
 \begin{center}
 \includegraphics[height=5.5cm,width=7.5cm]{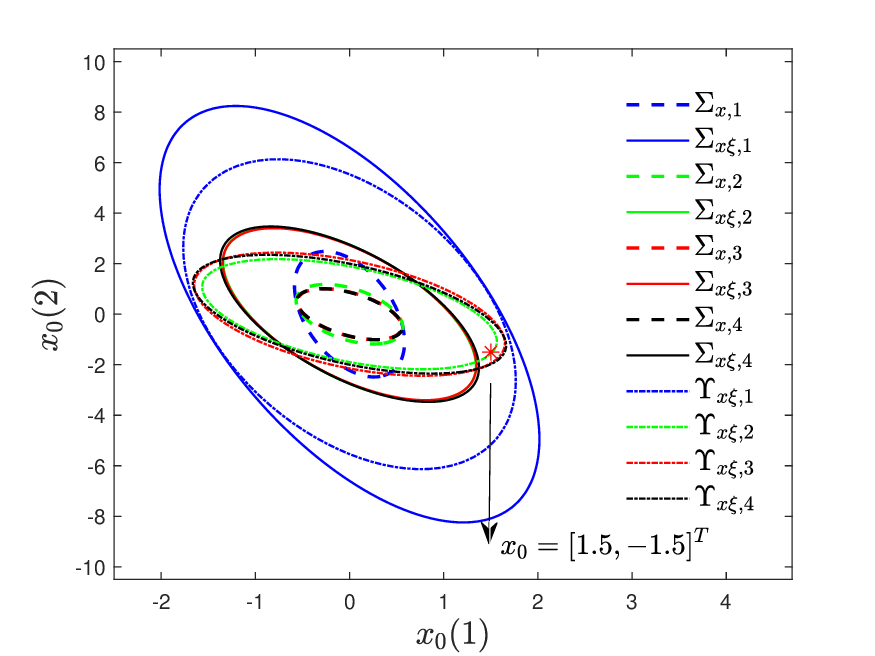}
 \end{center}
 \end{minipage}}
 \caption{Comparison of initial feasible region before and after expansion}\label{Fig.1}
 \end{figure}

\subsection{Solving On-line Part}
In this subsection, the benefits of the proposed DPO-MPC are illustrated through comparative simulations. To demonstrate the comparison with EMPC and online MPC, the average performance from 100 different experiments is utilized due to random jumps in FMJSs. Algorithm 2 can be effectively executed using the \emph{yalmip.master} toolbox on the MATLAB R2016a platform featuring an Intel(R) Core(TM) i5-3470 CPU @3.20GHz.

From TABLE \ref{tab:1}, it is evident that the online computational cost of DPO-MPC is markedly reduced in contrast to the EMPC and Online robust MPC (RMPC) strategies. While \textbf{OP5} requires online solving, its constraints are considerably fewer compared to the RMPC approach. Furthermore, unlike EMPC, the online \textbf{OP5} constraints involve only a single perturbation $\eta_s$ rather than a sequence, thereby reducing the computational burden to some extent. The simulation outcomes are depicted in Fig.~\ref{Fig.4}-Fig.~\ref{Fig.6}
\begin{table}[!ht]
\begin{center}
  \caption{Comparison of the average solver-time using \emph{yalmip} for DPO-MPC, EMPC, and online MPC}\label{tab:1}
 \begin{tabular}{m{1in}m{1in}m{1in}m{1in}}
\toprule
\toprule
\multicolumn{1}{c}{Method} &\multicolumn{1}{c}{DPO-MPC }&\multicolumn{1}{c}{EMPC}&\multicolumn{1}{c}{Online-MPC}\\
\midrule
\multicolumn{1}{c}{Time(s)}&\multicolumn{1}{c}{0.1395}&\multicolumn{1}{c}{0.5660}&\multicolumn{1}{c}{0.8974}\\
\bottomrule
\bottomrule
\end{tabular}
\end{center}
\end{table}
 \begin{figure}[!ht]
 \centering
 {  \begin{minipage}[t]{0.4\textwidth}
 \begin{center}
 \includegraphics[height=5.5cm,width=7.5cm]{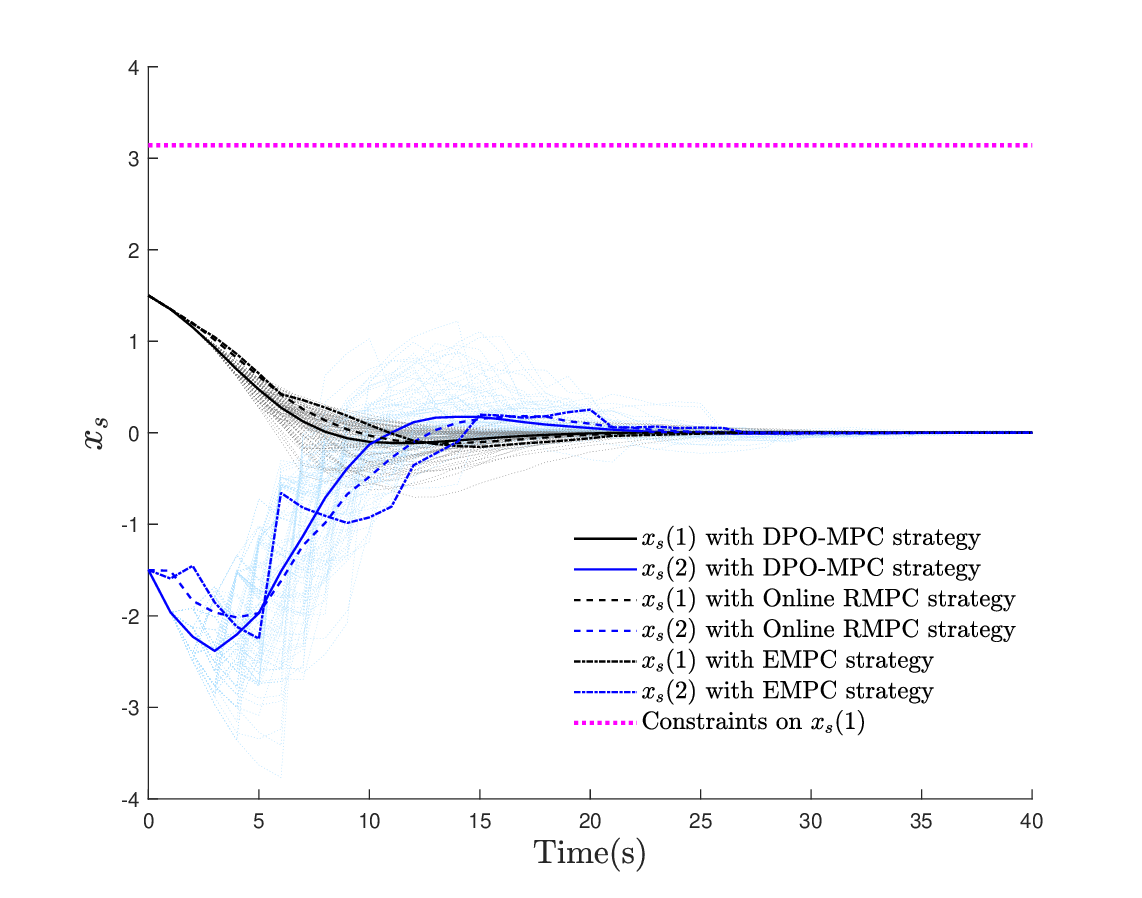}
 \end{center}
 \end{minipage}}
 \caption{The system state evolution $x_s$ with three different MPC strategy. (Closed-loop evolution for 100 times different experiments.)}\label{Fig.4}
 \end{figure}

\begin{figure}[!ht]
 \centering
 {  \begin{minipage}[t]{0.4\textwidth}
 \begin{center}
 \includegraphics[height=5.5cm,width=7.5cm]{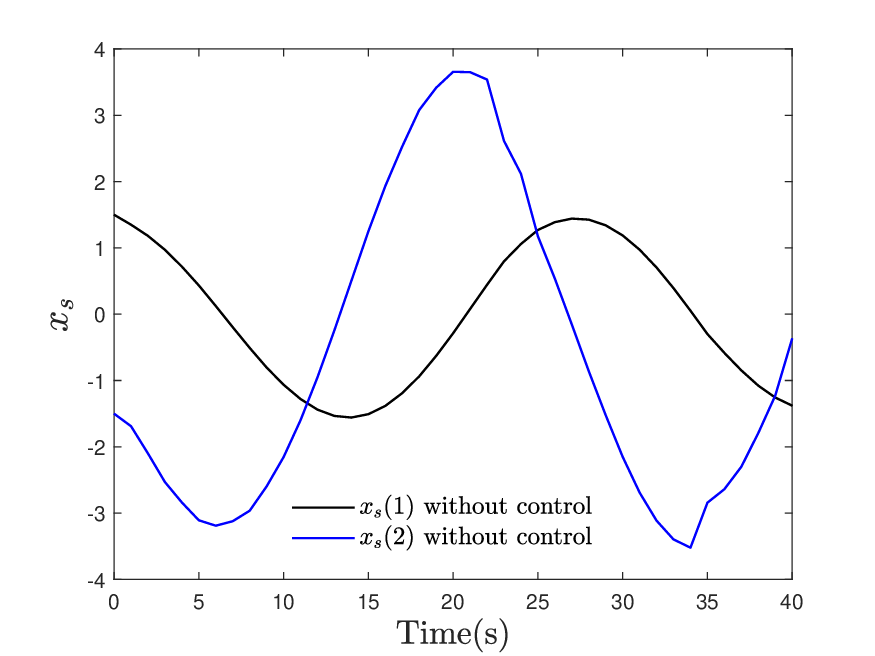}
 \end{center}
 \end{minipage}}
 \caption{The system states evolution $x_s$ without control.}\label{Fig.5}
 \end{figure}
  \begin{figure}[!ht]
 \centering
 {  \begin{minipage}[t]{0.4\textwidth}
 \begin{center}
 \includegraphics[height=5.5cm,width=7.5cm]{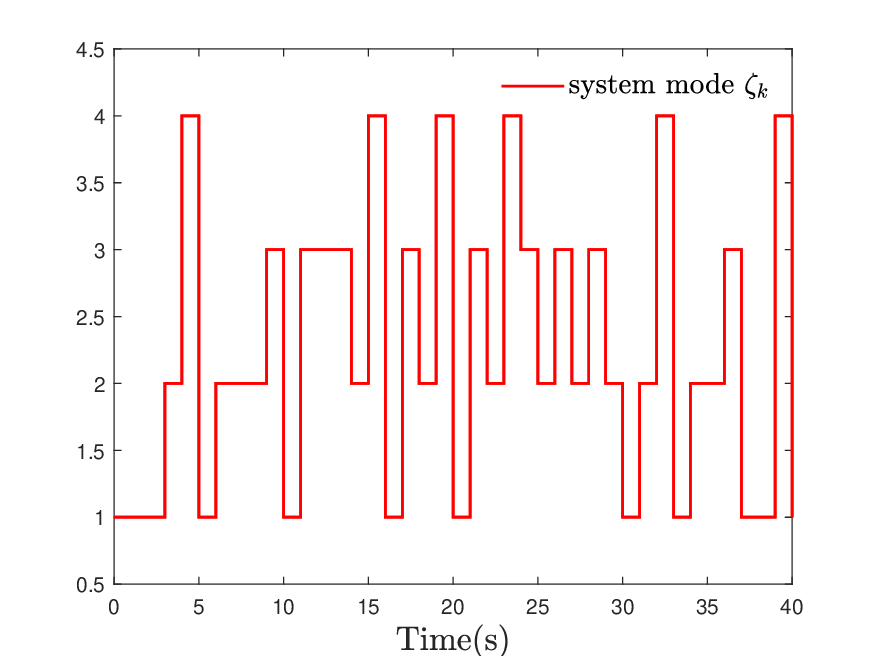}
 \end{center}
 \end{minipage}}
 \caption{A possible sequence of system modes}\label{Fig.2}
 \end{figure}
 \begin{figure}[!ht]
 \centering
 {  \begin{minipage}[t]{0.4\textwidth}
 \begin{center}
 \includegraphics[height=5.5cm,width=7.5cm]{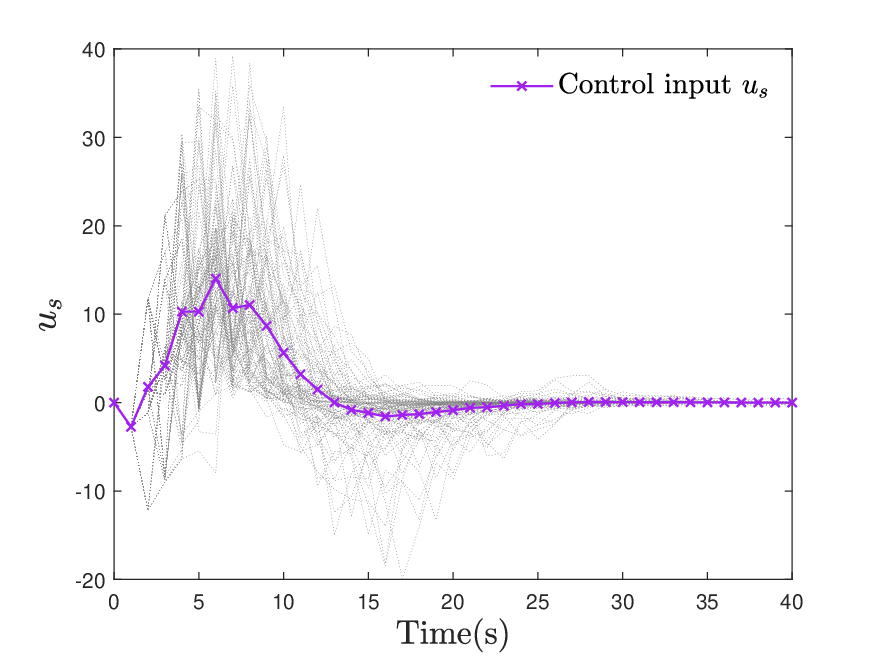}
 \end{center}
 \end{minipage}}
 \caption{The evolution of control inputs. (Closed-loop evolution for 100 times different experiments.)}\label{Fig.6}
 \end{figure}
\subsection{Discussion and Analysis}
From Fig.~\ref{Fig.4}, it can be concluded that the control performances of the compared strategies exhibit similarity. However, a detailed examination reveals that the control effectiveness of the DPO-MPC approach slightly surpasses that of the EMPC method, approaching parity with the online RMPC strategy. Thus, this paper validates the efficacy of the proposed methodology. Furthermore, to demonstrate the effectiveness of the MPC strategy, we employ a divergent open-loop system, depicted in Fig.~\ref{Fig.5}. The sequence of system modes is illustrated in Fig.~\ref{Fig.2}. Finally, Fig.~\ref{Fig.6} portrays the trajectory of the system's control input.
\section{Conclusion}\label{sec:6}
This study has investigated the DPO-MPC problem for a discrete-time class of FMJSs featuring hard constraints. By utilizing the IPM approach, a suite of mode-dependent fuzzy predictive controllers is devised to ensure system stability. The optimized predictive dynamics MPC strategy significantly expands the initial feasible region and reduces online computational overhead, enhancing algorithm practicality. The design methodology is synthesized through an ``off-line to online'' approach, establishing a comprehensive framework for analyzing algorithm feasibility and mean-square stability in the underlying MJS. Theoretical findings are validated via a single-link robot arm system. Future research directions include extending these results to encompass more complex systems with advanced network-induced phenomena, as explored in \citep{Zhao.D20_A}.
\section*{Acknowledgments}
This work was supported in part by the China Postdoctoral Science Foundation under Grants 2022TQ0208, 2023M732226.

\end{document}